\RequirePackage{fix-cm}
\documentclass[twocolumn]{svjour3}       
\smartqed  
\usepackage{graphicx}
\usepackage{amsmath}   
\usepackage{amssymb}
\usepackage{mathrsfs}
 \usepackage{hyperref}
 \usepackage{url} 
\usepackage{enumitem}
\usepackage{bbm,tikz}
\usepackage{algorithm}
\setlist{nosep}
\usepackage{algpseudocode}
\algrenewcommand\algorithmicrequire{\textbf{Input:}}
\algrenewcommand\algorithmicensure{\textbf{Output:}}
\usepackage{amsmath,bm}
\usepackage{parallel,enumitem}
\usepackage{todonotes}
\usepackage{xspace,setspace}
\newcommand{\V}{\bm{V}}

\newcommand{\R}{\mathbb{R}}

\newcommand{\N}{\mathbb{N}}
\newcommand{\Q}{\mathbb{Q}}

\newcommand{\mL}{\mathcal{L}}
\newcommand{\CL}{\mathscr{C}}

\newcommand{\Fb}{\mathbf{F}}
\newcommand{\cb}{\mathbf{c}}
\newcommand{\xb}{\mathbf{x}}
\newcommand{\yb}{\mathbf{y}}
\newcommand{\ab}{\mathbf{a}}
\newcommand{\wb}{\mathbf{w}}
\newcommand{\hb}{h}

\newcommand{\bb}{\mathbf{b}}
\newcommand{\fb}{\mathbf{f}}
\newcommand{\gb}{\bm{g}}
\newcommand{\C}{\mathbb{C}}

\newcommand{\Cn}{\C^n}
\newcommand{\Span}{{\rm{span}}}

\def\1{{\mathbbm 1}}
\newcommand{\gbt}{\widetilde{\gb}}

\newcommand{\compose}{\textsf{Compose}\xspace}

\newcommand{\InRadical}{\textsf{InRadical}\xspace}

\newcommand{\InvariantSet}{\textsf{InvariantSet}\xspace}
\newcommand{\ComputeLoopsUniversal}{\textsf{ComputeLoopsUniversal}\xspace}
\newcommand{\ComputeLoopsLinearUniversal}{\textsf{ComputeLoopsLinearUniversal}\xspace}
\newcommand{\InvariantSetBranch}{\textsf{InvariantSetBranch}\xspace}
\newcommand{\solveHom}{\textsf{solveHom}\xspace}
\newcommand{\sol}{\textsf{sol}\xspace}

\newcommand{\GenerateLoops}{\textsf{GenerateLoops}\xspace}

\newcommand{\coef}{\textsf{coef}\xspace}

\newcommand{\programbox}[2][\linewidth]{
\begin{samepage}\normalfont
\vspace*{.5em}\hspace*{0.3cm}\fbox{
\hspace*{-0.3cm}\begin{minipage}{#1}
\vspace*{-.1em}\begin{algorithmic}
#2
\end{algorithmic}
\vspace*{-.2em}
\end{minipage}
}\\
\end{samepage}
}

\newcommand{\programboxappendix}[2][\linewidth]{
\begin{samepage}\normalfont
\vspace*{.5em}\hspace*{0.0cm}\fbox{
\hspace*{-0.3cm}\begin{minipage}{#1}
\vspace*{-.1em}\begin{algorithmic}
#2
\end{algorithmic}
\vspace*{-.2em}
\end{minipage}
}\\
\end{samepage}
}

\begin{document}

\title{From Affine to Polynomial: Synthesizing Loops with Branches via Algebraic Geometry
}


\author{Erdenebayar Bayarmagnai         \and
        Fatemeh Mohammadi \and R\'emi Pr\'ebet 
}


\institute{E. Bayarmagnai \at
             KU Leuven, Department of Computer Science\\
            \email{
erdenebayar.bayarmagnai@kuleuven.be}           
           \\
           F.Mohammadi \at
               KU Leuven, Department of Computer Science\\
              \email{fatemeh.mohammadi@kuleuven.be}
              \\
              R.Pr\'ebet \at
              Inria, CNRS, ENS de Lyon, Université Claude Bernard Lyon 1, LIP, UMR 5668, 69342, Lyon cedex 07, France\\
              \email{remi.prebet@ens-lyon.fr} 
}

\date{}

\maketitle

\begin{abstract}
Ensuring software correctness remains a fundamental challenge in formal program verification. One promising approach relies on finding polynomial invariants for loops. Polynomial invariants are properties of a program loop that hold before and after each iteration. Generating such invariants is a crucial task in loop analysis, but it is undecidable in the general case. Recently, an alternative approach to this problem has emerged, focusing on synthesizing loops from invariants. However, existing methods only synthesize affine loops without guard conditions from polynomial invariants. In this paper, we address a more general problem, allowing loops to have polynomial update maps with a given structure, inequations in the guard condition, and polynomial invariants of arbitrary form.

We use algebraic geometry tools to design and implement an algorithm that computes a finite set of polynomial equations whose solutions correspond to all nondeterministic branching loops satisfying the given  invariants. Furthermore, we introduce a new class of invariants for which we present a significantly more efficient algorithm.
In other words, we reduce the problem of synthesizing loops to find solutions of multivariate polynomial systems with rational entries.
This final step is handled in our software using an SMT solver. 
\end{abstract}

\section{Introduction}
Loop invariants are properties that hold before and after each iteration of a loop. They play a central role in automating program verification, which guarantees program correctness prior to execution. Various well-established methods rely on loop invariants for safety verification, including the Floyd–Hoare inductive assertion technique \cite{floyd1993assigning} and termination verification using standard ranking functions \cite{manna2012temporal}. When a loop invariant takes the form of a polynomial equation or inequality, it is referred to as a polynomial invariant. In this paper, we restrict our attention to invariants given by polynomial equations.

In this work, instead of generating polynomial invariants for a given loop, we address the reverse problem, that is synthesizing a loop that satisfies a specified set of polynomial invariants.

We consider polynomial loops $\mL(\ab, h, F)$ of the form
\begin{center}
  \programbox[0.73\linewidth]{
\State$\xb:=(x_{1},\ldots, x_n)\gets \ab:=(a_1,\ldots, a_n)$
\While{$h(\xb)\neq 0$}
\State $\begin{pmatrix}
x_1 \\
x_2 \\
\vdots \\
x_n
\end{pmatrix}
\xleftarrow{F}
\begin{pmatrix}
F_1\\
F_2\\
\vdots\\
F_n
\end{pmatrix}
$
\EndWhile
}\label{page:alg}
\end{center}
Here, $x_i$ are the program variables with initial values $a_i$, $h \in \C[\xb]$, and $F = (F_1, \ldots, F_n)$ is a sequence of polynomials in $\C[\xb]$. The condition $h(\xb) \neq 0$, called the \emph{guard}, is assumed to be a single inequation for simplicity. This is without loss of generality, since a guard of the form $h_1 \neq 0, \ldots, h_k \neq 0$ can be replaced by the product condition $h_1 \cdot \ldots \cdot h_k \neq 0$. When no guard is present, we write $\mL(\ab, 1, F)$, which corresponds to an infinite loop.  

In the following, we introduce some terminology from algebraic geometry. For further details, we refer the reader to \cite{cox2013ideals,kempf_1993,shafarevich1994basic}. We denote the field of complex numbers by $\mathbb{C}$. Throughout the paper, $\xb$ denotes the tuple of indeterminates $x_1, \ldots, x_n$, and $\C[\xb]$ the multivariate polynomial ring in these variables.  

\smallskip\noindent{\bf Ideals.} A polynomial ideal $I$ is a subset of $\mathbb{C}[\xb]$ that is closed under addition, $0\in I$ and for any $f\in \C[\xb]$ and $g\in I$, $fg\in I$. Given a subset $S$ of $\C[\xb]$, the ideal generated by $S$ is  
$$\langle S \rangle =\{a_1f_1+\ldots +a_mf_m\mid a_i\in \C[\xb], f_i\in S, m\in \N\}.$$
By Hilbert Basis Theorem~\cite[Theorem 4, Chap.~2]{cox2013ideals}, every ideal is generated by finitely many polynomials. For a subset $X\subset \Cn$, the defining ideal $I(X)$ of $X$ is the set of polynomials that vanish on $X$. The radical ideal of an ideal $I\subset \C[\xb]$ is defined as 
$$\sqrt{I}=\{f\in \C[\xb]| \,f^m\in I \text{ for some } m\in \N\}.$$

\smallskip\noindent{\bf Varieties.} For $S\subset \C[x]$, the algebraic variety $\V(S)$ is the common zero set of all polynomials in $S$. Moreover, $\V(S)=\V(\langle S\rangle)$ and so every algebraic variety is the vanishing locus of finitely many polynomials. 

\smallskip\noindent{\bf Polynomial maps.} A map $F:\C^n \longrightarrow \C^m$ is a polynomial map if there exist $f_1,\ldots, f_m\in \C[x_1,\ldots, x_m]$ such that $$F(x)=(f_1(x),\ldots, f_m(x))$$ for all $x\in \C^n$. For simplicity, we will identify polynomial maps and their corresponding polynomials.

In the following definition, $F^{(l)}(x)$ is defined recursively as $F^{(l)}(x) = F(F^{(l-1)}(x))$ for any $l > 1$ and $F^{(0)}(x) = x$ by convention.
\begin{definition}\label{def:PI}
A~polynomial~$g$ is an invariant of the loop $\mathcal{L}(\ab, h, F)$~if,~for~any~$m \in \mathbb{Z}_{\geq 0}$, either
$$g(F^{(m)}(\ab)) = 0,$$  
or there exists $m\in \mathbb{Z}_{\geq 0}$ such that 
$$g(F^{(m)}(\ab)) =h(F^{(m)}(\ab))=0,$$ 
and for every $0\leq l<m$
$$
    g(F^{(l)}(\ab)) = 0 \text{\, and\, }  h(F^{(l)}(\ab)) \neq 0.
$$
\end{definition}
\begin{definition}\label{def:InvId}
    Let $\mL(\ab, h, F)$ be a polynomial loop. The set of all polynomial invariants for $\mL(\ab, h, F)$  is called the \emph{invariant ideal} of $\mL$ and is denoted by $I_{\mL(\ab, h, F)}$. 
\end{definition}

The invariant ideal is indeed known (e.g. \cite{rodriguez2004automatic}) to be an ideal of $\C[x_1,\ldots, x_n]$ where $x_1,\ldots, x_n$ are the program variables. Consider polynomials $f_1,\ldots, f_n$ in $\C[x_1,\ldots, x_n]$, then we will denote by $$\Span\{f_1,\ldots, f_n\}$$ the vector space they generate.

\begin{definition}\label{def:Ffb}
    Let $\fb_1 =(f_{1,1},\ldots, f_{1,l_1}),\ldots, \fb_n =(f_{n,1},$ $\ldots, f_{n,l_n})$ and $(F_1,\ldots, F_n)$ be sequences of polynomials in $\C[\xb]$ 
    such that for every $i$, $$F_i =\displaystyle\sum_{j=1}^{l_i} b_{i,j}f_{i,j}$$ for some $b_{i,j}$ $ \in \C$. Let $\fb=(\fb_1\ldots, \fb_n)$ and define a  map $$F_{\fb,\bb}:\Cn \longrightarrow \Cn$$ with $F_{\fb,\bb}(x)=(F_1(x),\ldots, F_n(x))$.
    
\end{definition}

The object defined below is the primary focus of this paper.
We prove that it is an algebraic variety and compute the
polynomial equations that define it.
\begin{definition}\label{def:coefset}
    Using the notation of Definition~\ref{def:Ffb}, let  
    $h \in \C[\xb]$, $\gb = (g_1,\ldots,g_m)$ be a sequence of polynomials in $\C[\xb]$, and $\ab \in \C^n$.  
    The \emph{coefficient set} of the polynomial loop structured by $\fb$ with invariants $\gb$ is defined as
    \[
        \CL(\ab,h,\fb;\gb) \;=\; 
        \bigl\{\;\bb \in \C^{l_1+\cdots+l_n} \;\bigm|\; \gb \subset I_{\mL(\ab,h,F_{\fb,\bb})}\;\bigr\}.
    \]
\end{definition}

In simple words, $\CL(\ab,h, \fb; \gb)$ is the set of all coefficient vectors $\bb \in \C^{l_1+\cdots +l_n}$ such that all polynomials in $\gb$ are polynomial invariants of the following loop:
\begin{center}
  \programbox[0.55\linewidth]{
\State$\xb\gets\ab$
\While{$h(\xb) \neq 0$}
\State $\xb \gets F_{\fb,\bb}(\xb)
$
\EndWhile
}
\end{center}
In prior works~\cite{ISSAC2023Laura,hitarth_et_al:LIPIcs.STACS.2024.41,humenberger2022LoopSynthesis}, restricting computations to affine loops corresponds to choose $\fb_i = (1, x_1,$ $ \dotsc, x_n)$, for all $1\leq i \leq n$.

\medskip
\noindent\textbf{Related work.}
The computation of polynomial invariants for loops has been an active area of research over the past two decades \cite{bayarmagnai2024algebraic,hrushovski2018polynomial,karr1976affine,kovacs2008reasoning,kovacs2023algebra,rodriguez2004automatic,rodriguez2007automatic,rodriguez2007generating,de2017synthesizing}. However, computing invariant ideals is undecidable for general loops \cite{hrushovski2023strongest}. Consequently, efficient methods have been developed for restricted classes of loops, particularly those in which the assertions are linear or can be reduced to linear form.

The converse problem, namely synthesizing loops from given invariants, has received considerably less attention in the literature. 
Most existing work has focused on linear and affine loops. 
Synthesizing linear loops that satisfy given invariants was shown to be NP-hard in~\cite{ait2025simple}. 
These studies differ primarily in the types of invariants considered: linear invariants in~\cite{saurabh2010}, a single quadratic polynomial in~\cite{hitarth_et_al:LIPIcs.STACS.2024.41}, and pure difference binomials in~\cite{ISSAC2023Laura}. 
In contrast,~\cite{humenberger2022LoopSynthesis} considers general polynomial invariants but lacks completeness guarantees and does not address guard conditions. 

Higher-degree loops have been investigated in a more general framework in~\cite{synthesis2023algebro}, which also considers polynomial inequalities as input. 
However, the loops synthesized by that algorithm are restricted to cases where the input invariants are inductive, that is if an invariant holds after one iteration, it continues to hold for all subsequent iterations.

\medskip
\noindent\textbf{Our contributions.} 
In this work, we consider the following situation: the generation of polynomial loops with guards from arbitrary polynomial invariants (equalities). We take the first step towards this goal by generating a polynomial system whose solutions correspond exactly to loops with a given structure that satisfy the specified polynomial invariants.
We then discuss different strategies that can be used to solve this system. In practice, in most cases a satisfying solution is found using an SMT solver.

More precisely, in Section~\ref{section:preliminaries} we recall the definition of invariant sets and relevant results from \cite{bayarmagnaiIssac}. Section~\ref{section:generateloops} introduces a method for simultaneously identifying multiple polynomial invariants (Proposition~\ref{prop:invarianttest}) and applies it to show that the set $\CL(\ab, h, \fb; \gb)$, the collection of all coefficients of polynomial maps of loops satisfying the polynomial invariants $\gb$ with respect to $\fb$, forms an algebraic variety. We also present Algorithm~\ref{algo:generateloops}, which computes the defining polynomials of $\CL(\ab, h, \fb; \gb)$.
Section~\ref{sec:polysolve} discusses several strategies for solving the resulting polynomial systems and highlights their limitations. Finally, in Section~\ref{sec:implementation} we describe our implementation of the algorithms and report the experimental results. 
\medskip

\noindent\textbf{Extended version.}
This paper is an extended version of the paper~\cite{bayarmagnai2025beyond}, published in the Proceedings of the 14th ACM SIGPLAN International Workshop on the State of the Art in Program Analysis. It expands on the earlier work by providing extended results and introducing new contributions. Specifically, Section~\ref{sec: branching} generalizes our previous work by addressing branching loops with nondeterministic conditional statements, whereas the earlier paper focused only on single-path loops. Section~\ref{sec: general invariants} introduces a new class of invariants, and we develop Algorithm~\ref{algo:universalgeneralcase} to compute polynomial equations that characterize all branching loops with a given structure that satisfy these invariants. In contrast to Algorithm~\ref{algo:generateloops}, Algorithm~\ref{algo:universalgeneralcase}  avoids costly Gr\"obner basis computation. In addition, when the given invariants are affine, Proposition~\ref{prop:affine space} shows that the set of coefficients of the update map for a loop satisfying these invariants forms an affine space. Algorithm~\ref{algo:generalinvariants} is then used to compute an affine basis for this space.

\section{Preliminaries}\label{section:preliminaries}
Here, we introduce the main objects of this paper and extend key results from \cite{bayarmagnaiIssac,bayarmagnai2024algebraic}.
\subsection{Invariant sets for a single polynomial map}
\begin{definition}\label{def:invariantset}
Let $F : \Cn \longrightarrow \Cn$ be a  polynomial map and $X\subset \Cn$.
The invariant set of $(F,X)$ is defined as:
\begin{center}
    $S_{(F,X)} = \{x \in X \mid \forall m\in \N, F^{(m)}(x) \in X\}.$
\end{center}
\end{definition}

The following proposition, adapted from \cite[Proposition 2.3]{bayarmagnaiIssac}, is the key result enabling the invariant sets to be computed algebraically. 
\begin{proposition}\label{prop:stabilization}
Let $X\subseteq\Cn$ be an algebraic variety and consider a polynomial map $F : \Cn \longrightarrow \Cn$. We define $$X_m=X\cap F^{-1}(X)\cap \ldots \cap F^{-m}(X)$$ for all $m \in \N$. Then, there exists $N\in \N$ such that 
$X_N=X_{N+1}$, and for any such index
$X_N = S_{(F,X)}$.
\end{proposition}
Building on Proposition~\ref{prop:stabilization}, the following algorithm computes invariant sets by means of an iterative outer approximation that converges in finitely many steps. The algorithm requires the following procedures:
\begin{itemize}
    \item \texttt{Compose}: given two sequences of polynomials $\gb=(g_1,\ldots, g_m)$ and $F=(F_1,\ldots, F_n)$ in $\Q[\xb]$, it returns  
    \[
        g_1\big(F_1(\xb), \ldots, F_n(\xb)\big), \;\ldots,\; g_m\big(F_1(\xb), \ldots, F_n(\xb)\big).
    \] 
    
    \item \texttt{InRadical}: given a sequence of polynomials $\widetilde{\gb}$ and a finite set $S \subset \Q[\xb]$, it outputs \texttt{True} if $\widetilde{\gb} \subset \sqrt{\langle S\rangle}$, and \texttt{False} otherwise.
\end{itemize}
These procedures are standard routines in symbolic computation and the latter one can be carried out efficiently using techniques such as Gr\"obner bases \cite{cox2013ideals}. For further details, we refer the reader to \cite{bayarmagnai2024algebraic}.
\begin{algorithm}[H]
\caption{\InvariantSet}\label{algo1}
\begin{algorithmic}[1]
\Require sequences $\gb$ and~$F = (F_1,\ldots, F_n)$  in $\mathbb{Q}[\xb]$.
\Ensure polynomials of common zero-set $S_{(F,{\V(\gb)})}$.\vspace*{.3em}
\State\label{Step:algo1.1} $S \gets \{\gb\};$
\State\label{Step:algo1.2} $\gbt \gets \compose(\gb,\,F);$
\While{ $\InRadical(\gbt ,\,S)==\texttt{False}$}
\label{Step:algo1.3}\State\label{Step:algo1.4} $S \gets S \cup \{\gbt\};$
\State\label{Step:algo1.5} $\gbt \gets \compose(\gbt,\,F);$
\EndWhile
\State \Return $S$;
\end{algorithmic}
\end{algorithm}
The termination and correctness of Algorithm~\ref{algo1} follow from Proposition~\ref{prop:stabilization} and \cite[Theorem~2.4]{bayarmagnaiIssac}.

\begin{example}
    Consider the polynomial map and the algebraic variety: 
    $$F = (2x_1 - 3x_2,\, x_1 + x_2) \text{ and } X = \V(x_1^2 - x_2^2 + x_1x_2),$$  
    respectively.
    On input $F$ and $g = x_1^2 - x_2^2 + x_1x_2$, Algorithm~\ref{algo1} computes the invariant set of $(F, X)$ through the following steps:  

    \begin{itemize}
        \item At Step~\ref{Step:algo1.1}, $S \gets \{g\}$.  
        At Step~\ref{Step:algo1.2}, $\widetilde{g}$ is set to  
        \[
            \compose(g,F) = 5x_1^2 - 15x_1x_2 + 5x_2^2.
        \]  

        \item At Step~\ref{Step:algo1.3}, a Gr\"obner basis of the ideal $\langle g,\, 1 - t\widetilde{g} \rangle$ is computed which makes the following decision  
        \[
            \InRadical(\widetilde{g}, S) = \texttt{False}.
        \]  

        \item At Step~\ref{Step:algo1.4}, $S \gets \{g,\, g\circ F\}$.  
        At Step~\ref{Step:algo1.5}, $\widetilde{g}$ is updated as  
        \[
            \compose(\widetilde{g},F) = -5x_1^2 - 35x_1x_2 + 95x_2^2.
        \]  

        \item This time, $\InRadical(\widetilde{g}, S) = \texttt{True}$, so the while-loop terminates.  
    \end{itemize}

    \noindent Therefore, the invariant set is $S_{(F,\V(g))}=\V(g,\, g\circ F)$.
\end{example}

\subsection{Invariants sets for multiple polynomial loops}
In many cases, we are interested not in a single polynomial loop but in systems with multiple loops or updates. This subsection generalizes the notion of invariant sets to the setting of multiple polynomial loops.

\begin{definition}\label{def: composemaps}
Let $F_1, \ldots, F_k : \mathbb{C}^n \to \mathbb{C}^n$ be polynomial maps.  
For any $m \in \mathbb{N}$ and any sequence $i_1, \ldots, i_m \in [k]$, define the polynomial map
$$F_{i_1, \ldots, i_m}(\xb) = F_{i_m} \big(F_{i_{m-1}}(\cdots F_{i_1}(\xb) )\big).$$ 
\end{definition}

The next definition extends Definition~\ref{def:invariantset} to the setting of multiple polynomial maps.  

\begin{definition}\label{def:invariantsetextend}
Consider the set $X \subseteq \Cn$ and the polynomial maps $F_1,\ldots,F_k : \Cn \to \Cn$. The invariant set of $((F_1,\ldots,F_k),X)$ is defined as
\begin{center}
    \scalebox{0.9}{$\displaystyle\Big\{x \in X \mid\forall m \geq 1, \forall  i_1,\ldots, \forall i_m\in [k],\; F_{i_m}\big(\ldots(F_{i_1}(x)\big) \in X \Big\}$}
\end{center} 
where $[k] = \{1,\ldots,k\}$. We note it $S_{((F_1,\ldots,F_k),X)}$.
\end{definition}

The next proposition introduce an algorithm for computing invariant sets of multiple polynomial maps. For more details, this is an adaptation of in~\cite[Theorem~3.14]{bayarmagnai2024algebraic}.

\begin{proposition}\label{prop: invariantsetbranch} 
    Let $F_1, \ldots, F_k$ and $\gb$ be sequences of polynomial in $\Q[\xb]$.
    Define \InvariantSetBranch the modified version of Algorithm~\ref{algo1}, obtained by replacing  
    \begin{itemize}\small 
        \item $\compose(\gb, F)$ with the tuple \vspace*{-.5em} 
        $$\big((\compose(\gb, F_1), \ldots, \compose(\gb, F_k)\big)\vspace*{-.5em}$$ 
        in Step~\ref{Step:algo1.2};, and \vspace*{1em} 
        \item $\compose(\widetilde{\gb}, F)$ with the tuple\vspace*{-.5em} 
        $$\big(\compose(\widetilde{\gb}, F_1), \ldots, \compose(\widetilde{\gb}, F_k)\big)\vspace*{-.5em}$$
        in Step~\ref{Step:algo1.5}.  
    \end{itemize}
   Then, on input the $F_i$ and $\gb$, \InvariantSetBranch computes polynomials whose common zero-set is the invariant set of $((F_1,\dotsc,F_k), \V(\gb))$.
   
\end{proposition}

\section{From loops to polynomial systems}\label{section:generateloops}

In this section, we present a method for identifying all (branching) loops that satisfy given polynomial invariants by formulating the problem as a polynomial system over invariant sets. This reduces loop synthesis to solving a polynomial system, which will be examined in detail in the next section.

\subsection{Single-path loops}
We first consider \emph{single-path loops}, i.e. loops defined by a single polynomial update map $F:\C^n \to \C^n$ with an inequation $h(x)\neq 0$ in the guard condition.  
This case provides the basic setting in which to study invariant sets, before extending the results to branching loops in Section~\ref{sec: branching}.

The following lemma, which is an immediate consequence of Definition~\ref{def:invariantset}, shows that invariant sets are preserved under intersections of algebraic varieties.
\begin{lemma}\label{prop:intersectioninv}
    Let $F:\C^n \to \C^n$ be a polynomial map, and let $X_1,\ldots,X_m \subseteq \C^n$ be algebraic varieties.  
    Then
  $$ S_{(F,\, X_1 \cap \cdots \cap X_m)} 
        \;=\; S_{(F,X_1)} \cap \cdots \cap S_{(F,X_m)}.$$
\end{lemma}

We now provide a necessary and sufficient condition for determining whether a given set of polynomials forms invariants of a loop. 
Furthermore, the next proposition extends~\cite[Proposition~2.7]{bayarmagnai2024algebraic}, where the statement was established for a single polynomial invariant.  

\begin{proposition}\label{prop:invarianttest}
    Let $h,g_1,\ldots,g_m \in \C[\xb]$, and let $z$ be a new indeterminate.  
    Define  
    $$
        X = \V(zg_1,\ldots,zg_m) \subset \C^{n+1}.
    $$  
    Let $F:\Cn \to \Cn$ be a polynomial map, and define  
    $$
        G_h(\xb,z) = (F(\xb),\, zh(\xb)).
    $$  
    Then, for $\ab \in \Cn$, one has $\gb \subseteq I_{\mL(\ab,h,F)}$ if and only if $(\ab,1) \in S_{(G_h,X)}$.
\end{proposition}

\begin{proof}
    For each $i \in \{1,\ldots,m\}$, let $X_i = \V(zg_i) \subset \C^{n+1}$ be an algebraic variety.  
    By~\cite[Proposition~2.7]{bayarmagnai2024algebraic}, we have  
    $$
        g_i \in I_{\mL(\ab,h,F)} 
        \;\;\Longleftrightarrow\;\; 
        (\ab,1) \in S_{(G_h,X_i)}.
    $$  
    Hence, $g_1,\ldots,g_m \in I_{\mL(\ab,h,F)}$ if and only if  
    $$
        (\ab,1) \in S_{(G_h,X_1)} \cap \cdots \cap S_{(G_h,X_m)}.
    $$  
    By Lemma~\ref{prop:intersectioninv}, it follows that  
    $$
        S_{(G_h,X)} = S_{(G_h,X_1)} \cap \cdots \cap S_{(G_h,X_m)}.
    $$  
    Hence, $g_1,\ldots,g_m \in I_{\mL(\ab,h,F)}$ if and only if $(\ab,1) \in S_{(G_h,X)}$.  \qed
\end{proof}

The next proposition provides a necessary and sufficient condition for loops with a given structure to satisfy specified polynomial invariants.

\begin{proposition}\label{prop:loopgenerator}
    Let $\fb_i =(f_{i,1},\ldots, f_{i,l_i})$ be a sequence of polynomials in $\C[\xb]$ for every $i\leq n$, and define $\fb = (\fb_1,\ldots,\fb_n)$.  
    Let $z, y_{i,1},\ldots,y_{i,l_i}$ be new indeterminates for every $i\leq n$, and let $h \in \C[\xb]$.   
    Let $H_{h}$ be the map
    \begin{center}
        \scalebox{.9}{$
        \begin{array}{ll}
        \displaystyle\C^{n+\sum_{j=1}^n l_j + 1} &\hspace*{-0cm}\longrightarrow \hspace*{.2cm} \C^{n+\sum_{j=1}^n l_j + 1}\\[.5em]
        (\xb,\yb,z)&\hspace*{-.8cm}\longmapsto
        \displaystyle\left(\sum_{i=1}^{l_1}y_{1,i}f_{1,i}(\xb),\ldots, \displaystyle\sum_{i=1}^{l_n}y_{n,i}f_{n,i}(\xb),\,\yb,\,zh(\xb)\right).
        \end{array}
        $}
    \end{center}
    
    Let $\gb = (g_1,\ldots,g_m)$ be a sequence of polynomials in $\C[\xb]$, viewed as elements of $\C[\xb,\yb,z]$, and define  
    $$
        X = \V(zg_1,\ldots,zg_m) \subset \C^{\,n+\sum_{j=1}^n l_j+1}.
    $$  
    Then, for any $\ab \in \C^n$,  
    $$
        \CL(\ab,h,\fb;\gb) \;=\;
        \Big\{\, \bb \in \C^{\,\sum_{j=1}^n l_j} \;\Big|\; (\ab,\bb,1) \in S_{(H_{h},X)} \Big\}.
    $$
\end{proposition}

\begin{proof}
    Write 
    $
        H_{h}(\xb,\yb,z) = \big(H(\xb,\yb),\, zh(\xb)\big),
    $
    where  
    $$
        H(\xb,\yb) =
        \left(
            \sum_{i=1}^{l_1} y_{1,i} f_{1,i}(\xb),\;\ldots,\;
            \sum_{i=1}^{l_n} y_{n,i} f_{n,i}(\xb),\;\yb
        \right)\!.
    $$  

    By Proposition~\ref{prop:invarianttest}, $g_1,\ldots,g_m$ are polynomial invariants of $\mL((\ab,\bb),h,H)$ if and only if $(\ab,\bb,1) $ is contained in $S_{(H_{h},X)}$.  
    Since the value of $\yb$ remains fixed at $\bb$ under iteration, for all $N \in \N$, we have  
    $$
        H^N(\xb,\bb) = \big(F_{\fb,\bb}^N(\xb),\,\bb\big).
    $$    
    Hence, $\gb \subset I_{\mL(\ab,h,F_{\fb,\bb})}$ if and only if $\gb \subset I_{\mL((\ab,\bb),h,H)}$, which in turn holds if and only if $(\ab,\bb,1) \in S_{(H_{h},X)}$.  \qed
\end{proof}

Since the invariant set is an algebraic variety, Proposition~\ref{prop:loopgenerator} implies that $\CL(\ab,h,\fb;\gb)$ is likewise an algebraic variety. Its defining equations are obtained by substituting $\xb = \ab$ and $z = 1$ into the system that defines the invariant set.

The following algorithm computes a generating set of polynomials for $\CL(\ab,h,\fb;\gb)$.



\begin{algorithm}[H]
\caption{\GenerateLoops}\label{algo:generateloops}
\begin{algorithmic}[1]
\Require$ h, g_1,\ldots, g_m, f_{1,1},\ldots, f_{1, l_1},\ldots,f_{n,1},\ldots, f_{n, l_n}$ in $\Q[\xb]$ and $\ab\in \Q^n$.
\Ensure A set of polynomials $\{P_1,\ldots, P_s\}$ such that $\CL(\ab,h,\fb;g)$ is $\V(P_1,\ldots, P_s)$\vspace*{.5em}
\State{\small $H_{h} \gets \left(\displaystyle\sum_{i=1}^{l_1}y_{1,i}f_{1,i}(\xb),\ldots, \displaystyle\sum_{i=1}^{l_n}y_{n,i}f_{n,i}(\xb),\yb, zh(\xb)\right)$}\vspace*{.5em}
\State\label{Step:A1.3}$\{Q_1,\ldots, Q_s\} \gets \InvariantSet((zg_1,\ldots, zg_m), H_{h});$\vspace*{.5em}
\State\label{Step:A1.4}$\{P_1(\yb),\ldots, P_s(\yb)\}\gets \{Q_1(\ab,\yb,1),\ldots, Q_s(\ab,\yb,1)\} $\vspace*{.5em}
\State\Return $\{P_1,\ldots, P_s\}$;
\end{algorithmic}
\end{algorithm}
\begin{theorem}
    Let $\ab \in \Q^n$, $h \in \C[\xb]$, $\gb=(g_1,\ldots,g_m)$, and $\fb_1=(f_{1,1},\ldots,f_{1,l_1}), \ldots, \fb_n=(f_{n,1},\ldots,f_{n,l_n})$ be sequences of polynomials in $\Q[\xb]$. 
    Set $\fb=(\fb_1,\ldots,\fb_n)$.  
    Given input $(\ab,h,\fb,\gb)$, Algorithm~\ref{algo:generateloops} outputs a set of polynomials whose vanishing locus is $\CL(\ab,h,\fb;\gb)$.
\end{theorem}
\begin{proof}
    We prove the correctness of Algorithm~\ref{algo:generateloops}.  
    Let $X = \V(zg_1,\ldots,zg_m) \subset \C^{\,n + l_1 + \cdots + l_n + 1}$.  
    At Step~\ref{Step:A1.3}, by \cite[Theorem~2.4]{bayarmagnaiIssac}, \InvariantSet outputs  
    $$
        \{Q_1(\xb,\yb,z),\ldots,Q_s(\xb,\yb,z)\},
    $$  
    such that $S_{(H_{h},X)}$ is the common zero-set of $\{Q_1,\ldots,Q_s\}$.  
    Hence, by Proposition~\ref{prop:loopgenerator}, $g_1,\ldots,g_m$ are polynomial invariants of $\mL(\ab,h,F_{\fb,\bb})$ if and only if  
    $$
        Q_1(\ab,\bb,1) = \cdots = Q_s(\ab,\bb,1) = 0.
    $$  
    Consequently, $g_1,\ldots,g_m \in I_{\mL(\ab,h,F_{\fb,\bb})}$ if and only if  
    $$
        P_1(\bb) = \cdots = P_s(\bb) = 0,
    $$  
    which establishes the correctness of Algorithm~\ref{algo:generateloops}.\qed
\end{proof}

In the examples below, to distinguish program variables from the coefficients of the update maps, we denote the coefficients by $\lambda$ instead of $y$.

\begin{example}\label{exa:algo2}
    Consider the polynomial map
    $$
        F(x_1,x_2,x_3) = (\lambda_1 x_1^3 + \lambda_2 x_2^2,\;\lambda_3 x_1 + \lambda_4 x_2^2,\;\lambda_5 x_1),
    $$
    and the polynomial invariants 
    $g_1 = x_2^2 - x_1$ and $g_2 = x_3^3 + 2x_2^2 - x_1$.  
    Our goal is to determine all loops of this form with precondition $(x_1,x_2,x_3)=(1,1,-1)$ and postcondition $(x_2^2 - x_1 = 0,\; x_3^3 + 2x_2^2 - x_1 = 0)$:

    \programboxappendix[0.6\linewidth]{
    \State $(x_1,x_2,x_3)\gets(1,1,-1)$
    \While{true}
    \State $\begin{pmatrix}
    x_1 \\
    x_2 \\
    x_3
    \end{pmatrix}
    \xleftarrow{F}
    \begin{pmatrix}
    \lambda_1 x_1^3 + \lambda_2 x_2^2 \\
    \lambda_3 x_1 + \lambda_4 x_2^2 \\
    \lambda_5 x_1
    \end{pmatrix}$
    \EndWhile
    }
    \smallskip
    
    \noindent Let $F = (F_1,F_2,F_3)$ and $\gb = \{g_1,g_2\}$, and define
    $$
        \fb = \{\{x_1^3, x_2^2\},\;\{x_1, x_2^2\},\;\{x_1\}\}.
    $$  
    From the loop structure we have
    $$
        F_1 \in \Span\{x_1^3, x_2^2\},\,
        F_2 \in \Span\{x_1, x_2^2\},\,
        F_3 \in \Span\{x_1\}.
    $$  
    Thus, the input to Algorithm~\ref{algo:generateloops} is $(\{1\},\gb,\fb,\{1,1,-1\})$.  
    Let $X = \V(zg_1,zg_2) \subset \C^9$ and define the map
    $$
        H_{h}(\xb,\yb,z) =
        (y_1 x_1^3 + y_2 x_2^2,\;
        y_3 x_1 + y_4 x_2^2,\;
        y_5 x_1,\;
        \yb,\; z).
    $$  
    
    At Step~\ref{Step:A1.3}, given input $H_{h}$ and $X$, \texttt{InvariantSet} outputs polynomials
    $$
        \{Q_1(\xb,\yb,z),\ldots,Q_6(\xb,\yb,z)\}.
    $$  
    whose vanishing locus is the invariant set $S_{(H_{h},X)}$.
    Substituting the initial values of the loop into these polynomials yields four nonzero polynomials $P_1(\yb),\ldots$, $P_4(\yb)$, whose vanishing locus is $\CL((1,1,-1),1,\fb;\gb)$:
    \begin{center}\footnotesize
    $
        \bm{P_1} = (y_3+y_4)^2 - y_1 - y_2, 
        \quad
        \bm{P_2} = y_5^3 + 2(y_3+y_4)^2 - y_1 - y_2,
    $\vspace*{-1em}
    \end{center}
    \begin{flushleft}\footnotesize \setstretch{1.3}
    $
    \bm{P_3} = 2y_3^4y_4^2 + 8y_3^3y_4^3 + 12y_3^2y_4^4 + 8y_3y_4^5 + 2y_4^6 
    + y_1^3y_5^3 + 3y_1^2y_2y_5^3 + 3y_1y_2^2y_5^3 + y_2^3y_5^3 
    + 4y_1y_3^3y_4 + 4y_2y_3^3y_4 + 8y_1y_3^2y_4^2 + 8y_2y_3^2y_4^2 
    + 4y_1y_3y_4^3 + 4y_2y_3y_4^3 
    - y_1^4 - 3y_1^3y_2 - 3y_1^2y_2^2 - y_1y_2^3 
    + 2y_1^2y_3^2 + 4y_1y_2y_3^2 + 2y_2^2y_3^2 
    - y_2y_3^2 - 2y_2y_3y_4 - y_2y_4^2,
    $\vspace*{-.5em}
    \end{flushleft}
    \begin{flushleft}\footnotesize \setstretch{1.3}
    $
   \bm{P_4} = y_3^4y_4^2 + 4y_3^3y_4^3 + 6y_3^2y_4^4 + 4y_3y_4^5 + y_4^6 
    + 2y_1y_3^3y_4 + 2y_2y_3^3y_4 
    + 4y_1y_3^2y_4^2 + 4y_2y_3^2y_4^2 
    + 2y_1y_3y_4^3 + 2y_2y_3y_4^3 
    - y_1^4 - 3y_1^3y_2 - 3y_1^2y_2^2 - y_1y_2^3 
    + y_1^2y_3^2 + 2y_1y_2y_3^2 + y_2^2y_3^2 
    - y_2y_3^2 - 2y_2y_3y_4 - y_2y_4^2.
    $
    \end{flushleft}
\end{example}
\begin{remark}\label{rem:extendIV}
This algorithm can be straightforwardly extended to select an initial value in addition to the update map of the loop. Indeed, when no initial values are specified, Step~\ref{Step:A1.3} of Algorithm~\ref{algo:generateloops} produces a sequence of polynomials 
$$(Q_1,\ldots,Q_s) \subset \Q[\xb,\yb,z]$$ 
such that $\bb \in \CL(\ab,h,\fb;\gb)$ if and only if 
$$Q_1(\ab,\bb,1)=\cdots=Q_{s}(\ab,\bb,1)=0.$$ 
Consequently, Algorithm~\ref{algo:generateloops} applies even in the absence of initial values.
\end{remark}

\subsection{Generalization to branching loops}\label{sec: branching}

In this section, we extend the results on single-path loops to \emph{branching loops} 
$\mathcal{L}(\ab,h,(F_1,\ldots,F_k))$, which are defined by a nondeterministic conditional statement with $k$ branches:

\programbox[0.8\linewidth]{
\State $(x_{1},\ldots,x_n) := (a_1,\ldots,a_n)$
\While{$h \neq 0$}
    \If{$\ast$}
        \State $(x_{1},\ldots,x_n) \gets F_{1}(x_{1},\ldots,x_n)$\\\vspace*{-.5em}
        \hspace{5mm}$\textbf{\vdots}$
    \ElsIf{$\ast$}
        \State $(x_{1},\ldots,x_n) \gets F_{i}(x_{1},\ldots,x_n)$\\\vspace*{-.5em}
        \hspace{5mm}$\textbf{\vdots}$
    \Else
        \State $(x_{1},\ldots,x_n) \gets F_{k}(x_{1},\ldots,x_n)$
    \EndIf
\EndWhile
}

\noindent Here $\ab = (a_1,\ldots,a_n) \in \C^n$, $h \in \C[\xb]$, and  
$$F_1,\ldots,F_k : \C^n \to \C^n$$  
are polynomial maps. When no guard $h$ is specified, we simply write $\mathcal{L}(\ab,1,(F_1,\ldots,F_k))$.  

The following lemma, an immediate consequence of Definition~\ref{def:invariantsetextend}, shows that invariant sets remain closed under intersections in the branching case.  

\begin{lemma}\label{prop:extension intersectioninv}
    Let $F_1,\ldots,F_k : \C^n \to \C^n$ be polynomial maps and let $X_1,\ldots,X_m \subseteq \C^n$ be algebraic varieties.  
    Denote $\Fb = (F_1,\ldots,F_k)$. Then
    $$
        S_{(\Fb,\,X_1 \cap \cdots \cap X_m)} 
        \;=\; S_{(\Fb,X_1)} \cap \cdots \cap S_{(\Fb,X_m)}.
    $$
\end{lemma}

The invariant ideal $I_{\mathcal{L}(\ab,h,(F_1,\ldots,F_k))}$ of a branching loop is defined in~\cite[Definition~3.11]{bayarmagnai2024algebraic} as the set of all polynomial invariants of the loop.  
The following proposition extends Proposition~\ref{prop:invarianttest}; the proof is the same, except that~\cite[Proposition~3.15]{bayarmagnai2024algebraic} is invoked in place of~\cite[Proposition~2.7]{bayarmagnai2024algebraic}.

\begin{proposition}\label{prop:extension of invarianttest}
    Let $F_1,\ldots,F_k : \C^n \to \C^n$ be polynomial maps and let $h,g_1,\ldots,g_m \in \C[\xb]$.  
    Let $z$ be a new indeterminate, and define the variety
    $$
        X = \V(zg_1,\ldots,zg_m) \subset \C^{n+1}.
    $$  
    For each $i \leq k$, define the polynomial map
    $$
        G_{i,h}(\xb,z) = \big(F_i(\xb),\, zh(\xb)\big).
    $$  
    Then
    $$
        \gb \subseteq I_{\mathcal{L}(\ab,h,(F_1,\ldots,F_k))}
        \; \Longleftrightarrow \;        (\ab,1) \in S_{((G_{1,h},\ldots,G_{k,h}),X)}.
    $$
\end{proposition}

We now define the set of branching loops, specified by multiple polynomial maps, an initial value, and a guard condition, that satisfy a given collection of polynomial invariants.  
The following definition extends Definition~\ref{def:coefset}.

\begin{definition}\label{def: extension of ceofset}
    For each $i \leq k$ and $j \leq n$, let $\fb_{i,j} = (f_{i,j,1},\ldots,f_{i,j,n_{i,j}})$ be a sequence of polynomials in $\C[\xb]$ and let $h \in \C[\xb]$.  
    Let $\gb = (g_1,\ldots,g_m)$ be a sequence of polynomials in $\C[\xb]$, and set $\fb_i = (\fb_{i,1},\ldots,\fb_{i,n})$ for each $i \leq k$.  
    
    The set of branching loops structured by $(\fb_1,\ldots,\fb_k)$ and satisfying the invariants $\gb$ is the coefficient set $\CL(h,(\fb_1,\ldots,\fb_k);\gb)$ defined by 
\begin{align*}
     \Big\{(\ab,\bb_1,\ldots,\bb_k) \in \C^M \;\Big|\;
    \gb \subseteq I_{\mL(\ab,h,(F_{\fb_1,\bb_1},\ldots,F_{\fb_k,\bb_k}))} \Big\},
\end{align*}
    where 
    $
        \textstyle{M = n + \sum_{i=1}^k \sum_{j=1}^n n_{i,j}.}
    $
\end{definition}

Thus, $\CL(h,(\fb_1,\ldots,\fb_k);\gb)$ consists of all choices of initial values and coefficients that yield branching loops satisfying the polynomial invariants $\gb$.  

Hence, to synthesize a branching loop structured by $(\fb_1,\ldots,\fb_k)$ that satisfies $\gb$, it suffices to identify a vector in $\CL(h,(\fb_1,\ldots,\fb_k);\gb)$.


We now show that $\CL(h,(\fb_1,\ldots,\fb_k);\gb)$ is an invariant set, analogous to Proposition~\ref{prop:loopgenerator} for the single-map case. 
The proof is similar and therefore omitted.  

\begin{proposition}\label{prop: extenstion loopgenerator}
    Using the notation of Definition~\ref{def: extension of ceofset}, let $z$ and $y_{i,j,l}$ be new indeterminates for every $i \leq k$, $j \leq n$, and $l \leq n_{i,j}$.  
    For each $i \leq k$, define the polynomial map 
        $H_{i,h} : \C^{M+1} \to \C^{M+1}$
    in $\C[\xb,\yb,z]$ such that $H_{i,h}(\xb,\yb,z) $ is equal to 
    $$
       \bigg(
            \sum_{l=1}^{n_{i,1}} y_{i,1,l} f_{i,1,l}(\xb),\;\ldots,\;
            \sum_{l=1}^{n_{i,n}} y_{i,n,l} f_{i,n,l}(\xb), 
            \yb,\; zh(\xb)
        \bigg).
    $$
    Define 
        $X = \V(zg_1,\ldots,zg_m) \subset \C^{M+1}$,
    and set $\bb = (\bb_1,\ldots,\bb_k)$.  
    Then $\CL(h,(\fb_1,\ldots,\fb_k);\gb)$ is the set 
    $$
       \Big\{\, (\ab,\bb) \in \C^M \;\Big|\;
        (\ab,\bb,1) \in S_{((H_{1,h},\ldots,H_{k,h}),X)} \,\Big\}.
    $$
\end{proposition}

Therefore, by computing a defining set of polynomials for the invariant set
$S_{((H_{1,h},\ldots,H_{k,h}),X)}$ and then substituting $z=1$,
we obtain defining equations for $\CL(h,(\fb_1,\ldots,\fb_k);\gb)$.
An algorithm for computing these equations in the branching (multi-map) setting
is provided by Proposition~\ref{prop: invariantsetbranch}.
The following example illustrates the procedure step by step, yielding a defining set of polynomials for
$\CL(h,(\fb_1,\ldots,\fb_k);\gb)$.

\begin{example}
    Consider branching loops with a nondeterministic conditional statement involving two branches, given by
        $$
        F_1(x_1,x_2) = (\lambda_1x_1 + \lambda_2,\, \lambda_3x_2), $$
        $$
        F_2(x_1,x_2) = (\lambda_4x_2 + \lambda_5,\, \lambda_6x_1),
    $$
    that satisfy the polynomial invariant $g = 2x_1 - x_2^2$.  
    Hence, we aim to compute all branching loops of the following form that preserve $g$:

    \programbox[0.7\linewidth]{
    \State $(x_1,x_2) := (a_1,a_2)$
    \While{$1 \neq 0$}
        \If{$\ast$}
            \State $(x_1,x_2) \xleftarrow{F_1} (\lambda_1x_1+\lambda_2,\, \lambda_3x_2)$
        \Else
            \State $(x_1,x_2) \xleftarrow{F_2} (\lambda_4x_2+\lambda_5,\, \lambda_6x_1)$
        \EndIf
    \EndWhile
    }

    \smallskip

    \noindent Let $\fb_1 = \{\{x_1,1\},\{x_2\}\}$ and $\fb_2 = \{\{x_2,1\},\{x_1\}\}$.  
    Then $\mL((a_1,a_2),1,(F_1,F_2))$ satisfies the polynomial invariant $g$ if and only if
    $$
        (a_1,a_2,\lambda_1,\ldots,\lambda_6) \in \CL(1,(\fb_1,\fb_2),g).
    $$  
    Let $X = \V(zg) \subset \C^9$, and define the polynomial maps
    $$
        H_{1,h}(\xb,\yb,z) = (y_1x_1+y_2,\, y_3x_2,\, \yb,\, z),
    $$
    $$
        H_{2,h}(\xb,\yb,z) = (y_4x_2+y_5,\, y_6x_1,\, \yb,\, z).
    $$  
    By Proposition~\ref{prop: extenstion loopgenerator},  
    $
        (a_1,a_2,\lambda_1,\ldots,\lambda_6) \in \CL(1,(\fb_1,\fb_2),g)$  if and only if 
        $(a_1,a_2,\lambda_1,\ldots,\lambda_6,1) \in S_{(H_{1,h},H_{2,h}),X}.
    $
    On input $H_{1,h}, H_{2,h}$ and $X$,  
    \InvariantSetBranch outputs $31$ polynomials
    $$
        \{Q_1(\xb,\yb,z),\ldots,Q_{31}(\xb,\yb,z)\},
    $$
    whose vanishing locus is $S_{(H_{1,h},H_{2,h}),X}$.  
    Therefore, the vanishing locus of the equations $$Q_1(\xb,\yb,1)=0,\ldots,Q_{31}(\xb,\yb,1)=0,$$
    is $\CL(1,(\fb_1,\fb_2),g)$.
    Below we present six sample polynomials among these:
\begin{center}\scalebox{0.8}{\parbox{\linewidth}{
$2x_1-x_2^2;\quad  2y_1x_1+2y_2-y_3^2x_2^2; \quad 2y_4x_2+2y_5-y_6^2x_1^2;$\\[.5em]
$2y_1^2x_1+2y_1y_2+2y_2-y_3^4x_2^2;\, 2y_1y_4x_2+2y_1y_5+2y_2-y_3^2y_6^2x_1^2$\\[.5em]
$-y_1^2y_6^2x_1^2-2y_1y_2y_6^2x_1-y_2^2y_6^2+2y_3y_4x_2+2y_5.$
}}
\end{center}\vspace*{0em}
\end{example}

\begin{remark}
Similarly to Remark~\ref{rem:extendIV}, we can extend this approach (and in particular the previous algorithm) to also synthesize guard inequations.
Indeed, using the notation of Definition~\ref{def: extension of ceofset}, let $z$, $y_{i,j,l}$ and $w_1,\ldots, w_r$ be new indeterminates for every $i \leq k$, $j \leq n$, and $l \leq n_{i,j}$.  
Let $\hb=(h_1,\ldots,h_r)$ be a sequence of polynomials in $\Q[\xb]$ given as input. We look for guard polynomial is of the form $$h(\xb,\wb)=w_1h_1(\xb)+\cdots+w_rh_r(\xb)$$ where $\wb = (w_1,\ldots, w_r)$.  For each $i\leq k$, define polynomial maps 
    $H_{i,\hb}(\xb,\yb,z,\wb)$
    \begin{center}
        \scalebox{.9}{$\displaystyle
       \bigg(
            \sum_{l=1}^{n_{i,1}} y_{i,1,l} f_{i,1,l}(\xb),\ldots,
            \sum_{l=1}^{n_{i,n}} y_{i,n,l} f_{i,n,l}(\xb), 
            \yb, zh(\xb,\wb),\wb
        \bigg).$}
    \end{center}
      Define 
        $X = \V(zg_1,\ldots,zg_m)$.
    By computing defining equations for the invariant set  $$S_{((H_{1,\hb},\ldots,H_{k,\hb}),X)}$$ and then substituting $z=1$, 
we obtain equations $$P_1(\xb,\yb,\wb)=0,\ldots, P_s(\xb,\yb,\wb)=0$$ such that $(\ab,\bb)$ is contained in $\CL(h(\xb,\cb),(\fb_1,\ldots,\fb_k);\gb)$ if and only if $P_1(\ab,\bb,\cb)=\cdots=P_s(\ab,\bb,\cb)=0$. 
\end{remark}

\subsection{Universally inductive invariants}\label{sec: general invariants}
Algorithm~\ref{algo:generateloops} relies on Gr\"obner basis computation, whose worst case complexity can be doubly exponential in the number of variables~\cite{mayr1982complexity}.  
Consequently, when the template of the update map contains many generators, Algorithm~\ref{algo:generateloops} may fail to terminate within a reasonable time.  
Moreover, the degrees of the polynomial equations computed by Algorithm~\ref{algo:generateloops} can grow exponentially.  

In this subsection, we address these issues by focusing on a special class of invariants.  
For this class, we rely only on linear algebra to generate polynomial equations characterizing loops that satisfy the given invariants, and the degrees of these equations are bounded above by the degrees of the invariants themselves.

\subsubsection{General case}
In this subsection, we consider polynomial invariants of the form $g(\xb)-g(\ab)$, which hold for every initial value $\ab \in \C^n$.  
Equivalently, $g(\xb)-g(\ab)$ is an invariant of $\mL(\ab,h,(F_1,\ldots,F_k))$ for all $\ab \in \C^n$.
\begin{definition}\label{def:general}
    Let $F_1,\ldots,F_k : \C^n \to \C^n$ be polynomial maps, and let $h \in \C[\xb]$.  
    A polynomial $g(\xb)$ is called a \emph{universally inductive invariant} of $\mL(h,(F_1,\ldots,F_k))$ if 
    $g(\xb)-g(\ab)$ is an invariant of $\mL(\ab,h,(F_1,\ldots,F_k))$ for every $\ab \in \C^n$.  
    The set of all universally inductive invariants of $\mL(h,(F_1,\ldots,F_k))$ is denoted by
        $UI_{\mL(h,(F_1,\ldots,F_k))}$.
\end{definition}

We now define the coefficient set that characterizes all branching loops structured by $(\fb_1,\ldots,\fb_k)$ which satisfy a prescribed collection of universally inductive invariants.

\begin{definition}
    Using the notation of Definition~\ref{def: extension of ceofset},  
    the set of all branching loops structured by $(\fb_1,\ldots,\fb_k)$ that satisfy the universally inductive invariants $\gb$ is defined by the coefficient set  $U\CL(h,(\fb_1,\ldots,\fb_k);\gb)$, which is
    $$ \Big\{ (\bb_1,\ldots,\bb_k) \in \C^M 
        \;\Big|\; \gb \subseteq UI_{\mL(h,(F_{\fb_1,\bb_1},\ldots,F_{\fb_k,\bb_k}))} \Big\},
    $$
    where
    $
        \textstyle{M = \sum_{i=1}^k \sum_{j=1}^n n_{i,j}.}
    $
\end{definition}

The following proposition, stated in~\cite[Corollary~4.4]{bayarmagnai2024algebraic}, provides a necessary and sufficient condition for a polynomial to be a universally inductive invariant.  
For completeness, we include a brief proof.
\begin{proposition}\label{prop: general}
    Let $F_1,\ldots,F_k : \C^n \to \C^n$ be polynomial maps, and let $h(\xb) \in \C[\xb] \setminus \{0\}$.  
    Then a polynomial $g(\xb)$ is a universally inductive invariant if and only if
    $$
        g(F_i(\xb)) = g(\xb) \quad \text{for every } i \leq k.
    $$
\end{proposition}

\begin{proof}
    ($\Rightarrow$) Suppose $g$ is a universally inductive invariant.  
    Let $\ab \notin \V(h)$. By the definition of invariant polynomials,
    $$
        g(F_i(\ab)) - g(\ab) = 0
    $$
    for every $i \leq k$. Hence $g(F_i(\ab)) = g(\ab)$ for all $\ab \notin \V(h)$.  
    Now let $\ab \in \V(h)$. Since $h(\xb) \neq 0$, there exists a sequence $\{\ab_n\}$ with $\ab_n \notin \V(h)$ such that $\lim_{n \to \infty} \ab_n = \ab$.  
    Because $g(F_i(\ab_n)) = g(\ab_n)$ and both $g$ and $F_i$ are continuous, we conclude that $g(F_i(\ab)) = g(\ab)$ for all $\ab \in \V(h)$.  
    Thus, $g(F_i(\ab)) = g(\ab)$ for every $\ab \in \C^n$, and so $g(F_i(\xb)) = g(\xb)$ as polynomials, for all $i \leq k$.  

    \medskip
    ($\Leftarrow$) Conversely, assume $g(F_i(\xb)) = g(\xb)$ for every $i \leq k$.  
    Then for any $i_1,\ldots,i_m \in [k]$ and all $\ab \in \C^n$, we obtain
    $$
        g(F_{i_m}(\cdots F_{i_1}(\ab))) - g(\ab) = 0.
    $$
    Hence $g(\xb) - g(\ab) \in I_{\mL(\ab,h,(F_1,\ldots,F_k))}$, and $g$ is a universally inductive invariant.\qed
\end{proof}
An invariant is called \emph{inductive} if, once it holds after a single iteration, it continues to hold for all subsequent iterations.  
By Proposition~\ref{prop: general}, we immediately obtain:

\begin{corollary}\label{cor:univ-inductive}
    Every universally inductive invariant is inductive.
\end{corollary}

Algorithm~\ref{algo:universalgeneralcase} takes as input a sequence $\gb$ of polynomials in $\Q[\xb]$ together with the structure $(\fb_1,\ldots,\fb_k)$ of polynomial maps of a loop, and outputs polynomials defining $U\CL(h,(\fb_1,\ldots,\fb_k);\gb)$.  
The auxiliary procedure \texttt{coef} takes as input a sequence of polynomials in $\Q[\xb,\yb]$ and returns the coefficients of the monomials in the variables $\xb$.  
The correctness of Algorithm~\ref{algo:universalgeneralcase} follows immediately from Proposition~\ref{prop: general}.

\begin{algorithm}[H]
\caption{\ComputeLoopsUniversal}\label{algo:universalgeneralcase}
\begin{algorithmic}[1]
\Require A sequence $\gb=(g_1,\ldots,g_m)$ of  polynomials,  and polynomials $f_{i,j,l} \in \Q[\xb]$ for all $i \leq k$, $j \leq n$, and $l \leq n_{i,j}$.
\Ensure A set of polynomials $C_1,\ldots, C_t$ such that $U\CL(h,(\fb_1,\ldots,\fb_k);\gb)$ is $\V(C_1,\ldots, C_t)$\vspace*{.5em}
\For{$i \in [k]$}
    \State $\!\!\!\!\!\! F_i \gets \bigg(
        \displaystyle\sum_{l=1}^{n_{i,1}} y_{i,1,l} f_{i,1,l}(\xb),\;
        \ldots,\;
        \displaystyle\sum_{l=1}^{n_{i,n}} y_{i,n,l} f_{i,n,l}(\xb)
    \bigg)$
\EndFor
\State\label{Step: algo3 4} $\{C_1,\ldots, C_t\}\gets \{\coef(g_i - g_i\circ F_j) \mid i \in [m],\, j \in [k]\}$ 
\State \Return $\{C_1,\ldots, C_t\}$
\end{algorithmic}
\end{algorithm}

As a concrete application, we demonstrate how our method can be used to synthesize loops that generate Markov triples.

\begin{example}[Markov triples~\cite{cassels1957introduction}] In this example, we synthesize a loop of the following form that generates a solution to the Markov equation
$$x_1^2+x_2^2+x_3^2-3x_1x_2x_3=0.$$
 \programbox[0.9\linewidth]{
    \State $(x_1,x_2,x_3) := (1,1,2)$
    \While{$1 \neq 0$}
        \If{$\ast$}
              \State $\begin{pmatrix}
    x_1 \\
    x_2 \\
    x_3
    \end{pmatrix}
    \xleftarrow{F_1}
    \begin{pmatrix}
    \lambda_1x_1+\lambda_2x_2 \\
    \lambda_3x_1x_2+\lambda_4x_3+\lambda_5 x_1^2 \\
    \lambda_6x_2+\lambda_7x_3
    \end{pmatrix}$
        \Else
        \State $\begin{pmatrix}
    x_1 \\
    x_2 \\
    x_3
    \end{pmatrix}
    \xleftarrow{F_2}
    \begin{pmatrix}
    \lambda_8x_1+\lambda_9x_2\\
    \lambda_{10}x_2x_3+\lambda_{11}x_1+\lambda_{12} x_2^2 \\
    \lambda_{13}x_2+\lambda_{14}x_3
    \end{pmatrix}$
        \EndIf
    \EndWhile
    }    
\end{example}
Let $g=x_1^2+x_2^2+x_3^2-3x_1x_2x_3$. We therefore consider branching loops $\mL((1,1,2),1,(F_1,F_2))$  that satisfy the universal polynomial invariant $g$ with structure 
{\scriptsize $$F_1(\xb)=(\lambda_1x_1+\lambda_2x_2,3x_1x_2+\lambda_4x_3+\lambda_5 x_1^2,\lambda_6x_2+\lambda_7x_3),\vspace*{-1.5em}$$
$$F_2(\xb)=(\lambda_8x_1+\lambda_9x_2,\lambda_{10}x_2x_3+\lambda_{11}x_1+\lambda_{12} x_2^2,\lambda_{13}x_2+\lambda_{14}x_3)$$}
 By Proposition~\ref{prop: general}, we have $$g(\xb)=g(F_1(\xb)), \quad g(\xb)=g(F_2(\xb)).$$
Hence, the coefficients of the monomials in the variables $\xb$ appearing in the polynomial pairs $(g(\xb), g(F_1(\xb)))$ and $(g(\xb), g(F_2(\xb)))$ must be equal. Hence, we obtain the following system of 32 polynomial equations:
  \[
        \begin{cases}
        3\lambda_1\lambda_3\lambda_6+3\lambda_2\lambda_5\lambda_6-\lambda_3^2= 0,\\
        3\lambda_8\lambda_{10}\lambda_{13}+3\lambda_{8}\lambda_{12}\lambda_{14}=0\\
            \lambda_5^2 = 0,\\
            \ldots\\  
            \lambda_{14}^2-1 = 0,\\
            
            \lambda_1^2 - 1 = 0.
        \end{cases}
    \]
Using the \texttt{Z3} solver, we obtain the following solution to the system above:
$$\lambda_2=\lambda_5=\lambda_7=\lambda_8=\lambda_{12}=\lambda_{13}=0, \quad \lambda_3=\lambda_{10}=3$$
$$\lambda_1=\lambda_4=\lambda_6=\lambda_9=\lambda_{11}=\lambda_{14}=-1$$
\subsubsection{Affine case}

We now restrict attention to affine universally inductive invariants.  
We show that $U\CL(h,(\fb_1,\ldots,\fb_k);\gb)$ is an affine space whenever $\gb$ consists of affine polynomials.  
Furthermore, we present an algorithm for computing an affine basis of $U\CL(h,(\fb_1,\ldots,\fb_k);\gb)$.  

\begin{proposition}\label{prop:affine space}
    Let $\gb = (g_1,\ldots,g_m)$ be affine polynomials.  
    Then the coefficient set $U\CL(h,(\fb_1,\ldots,\fb_k);\gb)$ is an affine space.
\end{proposition}

\begin{proof}
    Let $(\bb_1,\ldots,\bb_k) \in U\CL(h,(\fb_1,\ldots,\fb_k);\gb)$. Then, 
    by Proposition~\ref{prop: general}, the polynomials $g_1,\ldots,g_m$ are universally inductive invariants of 
    $\mL(h,(F_{\fb_1,\bb_1},\ldots,F_{\fb_k,\bb_k}))$ if and only if
    $$
        g_i(F_j(\xb)) = g_i(\xb) \quad \text{for every } i \leq m, \; j \leq k.
    $$  
    This condition requires that the coefficients of the monomials in $g_i(F_j(\xb))$ coincide with those in $g_i(\xb)$.  
    Since $g_1,\ldots,g_m$ are linear, these coefficients are themselves linear polynomials in 
    $\C[\bb_1,\ldots,\bb_k]$.  
     Consequently, \\ $U\CL(h,(\fb_1,\ldots,\fb_k);\gb)$ is the solution set of a system of linear equations, 
    and therefore an affine space. \qed
\end{proof}
Proposition~\ref{prop:affine space} shows that, in the affine case, the synthesis problem reduces to a problem of linear algebra.

\medskip
Let $\gb$ be a sequence of affine polynomials.  
The following algorithm computes a vector $v$ and a basis of a vector space $V$ such that  
$U\CL(h,(\fb_1,\ldots,\fb_k);\gb)$ is the affine space $v+V$.  

We now describe the auxiliary procedures used in Algorithm~\ref{algo:generalinvariants}:

\begin{itemize}
    \item \sol takes as input a sequence of affine polynomials $(l_1(\yb),\ldots,l_r(\yb))$ in $\Q[\yb]$ and returns a solution to the linear system  
    $$
        l_1(\yb) = 0,\;\ldots,\; l_r(\yb) = 0.
    $$  

    \item \solveHom takes as input a sequence of affine polynomials $(l_1(\yb),\ldots,l_r(\yb))$ in $\Q[\yb]$ and outputs a basis for the solution space of the corresponding homogeneous system, namely  
    $$
        l_1(\yb) = 0,\;\ldots,\; l_r(\yb) = 0.
    $$  
\end{itemize}

The correctness of Algorithm~\ref{algo:generalinvariants} follows immediately from Proposition~\ref{prop:affine space}.

\begin{algorithm}[H]
\caption{\ComputeLoopsLinearUniversal}\label{algo:generalinvariants}
\begin{algorithmic}[1]
\Require A sequence $\gb=(g_1,\ldots,g_m)$ of affine polynomials,  and polynomials $f_{i,j,l} \in \Q[\xb]$ for all $i \leq k$, $j \leq n$, and $l \leq n_{i,j}$.
\Ensure A vector $v \in \Q^M$ and a sequence of vectors $B$ such that 
        $U\CL(h,(\fb_1,\ldots,\fb_k);\gb) = v+V$, 
        where $V$ is the vector space generated by $B$.\vspace*{.5em}
\State $C\gets \ComputeLoopsUniversal(\gb,\fb_1,\ldots, \fb_k)$
\State\label{Step: algo3 5}$v \gets \sol(C)$
\State\label{Step: algo3 6}$B \gets \solveHom(C)$
\State \Return $\{v, B\}$
\end{algorithmic}
\end{algorithm}

\begin{example}
    Consider branching loops with a nondeterministic conditional statement involving two branches, defined by
    \[
        F_1(x_1,x_2) = (\lambda_1x_1^2 + \lambda_2x_1 + \lambda_3x_2,\; 
                        \lambda_4x_1^2 + \lambda_5x_1 + \lambda_6x_2),
    \]
    \[
        F_2(x_1,x_2) = (\lambda_7x_1 + \lambda_8x_2,\;
                        \lambda_9x_1 + \lambda_{10}x_2),
    \]
    such that $g = x_1 - x_2 + 1$ is a universally inductive invariant of $\mL(1,(F_1,F_2))$.  
    By Proposition~\ref{prop: general}, this requires
    \[
        g(\xb) - g(F_1(\xb)) = 0, 
        \qquad 
        g(\xb) - g(F_2(\xb)) = 0.
    \]
    Collecting the coefficients of the monomials in $\xb$ yields the linear system:
    \[
        \begin{cases}
            \lambda_4 - \lambda_1 = 0,\\
            \lambda_5 - \lambda_2 + 1 = 0,\\
            \lambda_6 - \lambda_3 - 1 = 0,\\
            \lambda_9 - \lambda_7 + 1 = 0,\\
            \lambda_{10} - \lambda_8 - 1 = 0.
        \end{cases}
    \]
    A particular solution is  
        $v = (0,1,-1,0,0,0,1,-1,0,0)$.
    Solving the associated homogeneous system yields a sequence $B$ of basis vectors:
    \[
        (1,0,0,1,0,0,0,0,0,0),\quad 
        (0,1,0,0,1,0,0,0,0,0),
    \]
    \[
        (0,0,1,0,0,1,0,0,0,0),\quad 
        (0,0,0,0,0,0,1,0,1,0),
    \]
    \[
        (0,0,0,0,0,0,0,1,0,1).
    \]
    Therefore, $g$ is a universally inductive invariant of the loop $\mL(1,(F_1,F_2))$ if and only if 
        $(\lambda_1,\ldots,\lambda_{10}) \in v + \Span(B)$.
\end{example}

\section{Solving Polynomial Systems}\label{sec:polysolve}

In the previous section, we saw that Algorithms~\ref{algo:generateloops} and~\ref{algo:universalgeneralcase} produce systems of multivariate polynomials whose solutions correspond precisely to loops with the prescribed structure and invariants.  
Thus, the problem of loop synthesis reduces to solving polynomial systems.

Since our goal is to obtain loops that can be represented finitely and exactly on a computer, we focus on \emph{rational solutions}, i.e., solutions with coefficients in $\Q$.  
Unlike the situation over $\C$ (where Hilbert's Nullstellensatz applies~\cite{Hilbert1893}, see also~\cite[Chap.~4, \S1]{cox2013ideals}) or over $\R$ (where the Tarski–Seidenberg theorem applies~\cite{Tar1951,Sei1954}, see e.g.~\cite{bpr2006}), deciding whether a polynomial equation has a rational solution is a long-standing open problem in number theory~\cite{Shla2011}.  
For integer solutions, this corresponds to Hilbert's Tenth Problem, which is known to be undecidable~\cite{Shla2011}.

Given these difficulties, we outline three main strategies for addressing the problem, although none is fully satisfactory or complete.

\subsection{Exploiting the structure of polynomial systems}

Although no general algorithm is known for computing (or even deciding the existence of) rational solutions to multivariate polynomial systems by exact methods, many practical instances can still be solved by leveraging structural properties.  
For example, consider the polynomial system obtained at the end of Example~\ref{exa:algo2}.  
While this case may seem simple or narrowly tailored, all benchmarks reported in Section~\ref{sec:implementation} can be handled in a similar way.  
More generally, decompositions of varieties with additional structure can sometimes be reduced to combinatorial problems; see, e.g.,~\cite{liwski2025solvable,liwski2025efficient}.

\begin{example}  
The variety $\V(P_1,\ldots,P_4)$ can be decomposed into finitely many irreducible components, which in turn split the system into simpler subsystems. This decomposition can be carried out using classical methods from computer algebra~\cite[Chap.~4, \S6]{cox2013ideals}; here we use the command \texttt{minimalPrimes} from Macaulay2~\cite{M2}.  
We obtain the following five irreducible components:

\medskip
\begin{itemize}
    \item[] $\mathbf{V}_1 = \V(y_5,\; y_3+y_4,\; y_1+y_2)$ 
    \item[] $\mathbf{V}_2 = \V(y_5+1,\; y_3+y_4-1,\; y_1+y_2-1)$
    \item[] $\mathbf{V}_3 = \V(y_5+1,\; y_3+y_4+1,\; y_1+y_2-1)$
    \item[] $\mathbf{V}_4 = \V(y_3+y_4-1,\; y_1+y_2-1,\; y_5^2-y_5+1)$
    \item[] $\mathbf{V}_5 = \V(y_3+y_4+1,\; y_1+y_2-1,\; y_5^2-y_5+1)$
\end{itemize}

\medskip

\noindent Thus, $\mL((1,1,-1),1,F)$ satisfies the polynomial invariants $\{g_1,g_2\}$ if and only if $(\lambda_1,\ldots,\lambda_5)$ lies in one of the above components.  
The simple structure of the defining polynomials makes the problem fully solvable.  

For the first three components, which involve only linear polynomials, the problem reduces to standard linear algebra, where efficient methods are available (see e.g., \cite{CS2011}). More specifically, for parameters $\mu_1,\mu_2 \in \Q$, we obtain the following loop maps:
\begin{center}
\scalebox{.9}{\parbox{1.1\linewidth}{
\begin{itemize}
    \item[] $\mathbf{V}_1$: $F_1(\xb) = \big(\mu_1(x_1^3 - x_2^2),\; \mu_2(x_1 - x_2^2),\; 0\big)$\vspace*{.5em}
    \item[] $\mathbf{V}_2$: $F_2(\xb) = \big(\mu_1x_1^3 + (1-\mu_1)x_2^2,\; \mu_2x_1 + (1-\mu_2)x_2^2,\; -1\big)$\vspace*{.5em}
    \item[] $\mathbf{V}_3$: $F_3(\xb) = \big(\mu_1x_1^3 - (1+\mu_1)x_2^2,\; \mu_2x_1 + (1-\mu_2)x_2^2,\; -1\big)$\vspace*{-.5em}
\end{itemize}
}}
\end{center}
\smallskip\noindent For the last two components, no rational solution exists, since the polynomial $y_5^2-y_5+1$ has no rational roots.
\end{example}

Another important case arises when the polynomial system computed by Algorithm~\ref{algo:generateloops} and Algorithm~\ref{algo:universalgeneralcase} (or one of its subsystems) has only finitely many solutions.  
Geometrically, this means that the associated variety (or one of its irreducible components) consists of finitely many points, i.e., it has \emph{dimension zero}.  
In this situation, such systems can be reduced (see, e.g., \cite[\S3]{DL2008}, \cite{Rou1999}) to polynomial systems with rational coefficients of the form
$$Q_1(x_1) = 0, \qquad x_i = Q_i(x_1) \;\;\text{for all } i \geq 2.$$

Finding all rational solutions then reduces to solving a univariate polynomial equation, which can be done algorithmically: either through efficient factoring~\cite[Theorem~15.21]{von2013modern} to extract the linear factors in $\Q[x_1]$, or via modular arithmetic-based methods~\cite{Lo1983}.  

Thus, whenever the system has finitely many solutions, all rational ones can be effectively computed.  
Practical implementations exist, such as \textsf{AlgebraicSolving.jl}\footnote{\url{https://algebraic-solving.github.io}}, which is based on the \texttt{msolve} library~\cite{BES2021}.  
Moreover, many benchmarks in the next section fall into this category, which we believe is common when the degree and support of $F$ are comparable.

\subsection{Numerical methods}

Numerical approaches to solving polynomial systems, such as \textsf{HomotopyContinuation.jl}~\cite{HomotopyContinuation}, provide powerful tools for approximating solutions.  
These methods rely on numerical algebraic geometry, in particular on homotopy continuation.  
Unlike symbolic techniques based on Gr\"obner bases or resultants, numerical methods can efficiently handle large and complex systems.  
When applied to the systems generated by Algorithm~\ref{algo:generateloops} and Algorithm~\ref{algo:universalgeneralcase}, they yield numerical solutions that can then be checked against nearby integers or rationals to identify valid exact solutions.
 

\begin{example}\label{ex4}
    Consider the polynomial system $\{P_1,P_2,P_3$, $P_4\}$ from Example~\ref{exa:algo2}.  
    Since it involves five variables and only four equations, it has either no solutions or infinitely many.  
    A common strategy is to augment the system with a random linear form with integer coefficients.  
    Using \textsf{HomotopyContinuation.jl}, we obtain 15 numerical real solutions to the augmented system, 12 of which correspond to valid integer solutions.
\end{example}

Despite their efficiency, numerical methods have limitations.  
They provide approximate solutions that may lack exact algebraic structure and require additional validation.  
They can also struggle with singular solutions, and homotopy continuation methods may occasionally miss solutions altogether.

\vspace*{-1em}
\subsection{Satisfiability Modulo Theories (SMT) solvers}

The SMT solvers decide the satisfiability of logical formulas that combine Boolean logic with background theories such as integer or real arithmetic.  
In particular, the existence of integer or rational solutions to a polynomial system can be expressed as
{\small \[\exists x_1 \ldots \exists x_n \; (f_1(x_1, \ldots, x_n) = 0) \wedge \ldots \wedge (f_m(x_1, \ldots, x_n) = 0).\]}  
SMT solvers benefit from automated reasoning, efficient decision procedures, and the ability to integrate multiple logical theories.  
However, they face difficulties with high-degree polynomials and, due to undecidability, cannot guarantee general solutions over the integers.  
Nevertheless, as shown in Section~\ref{sec:implementation}, the SMT solver \texttt{Z3} successfully finds integer solutions for most of the polynomial systems generated by Algorithms~\ref{algo:generateloops} and Algorithm~\ref{algo:universalgeneralcase}.

\section{Implementation and Experiments}\label{sec:implementation}

\subsection{Setup}
We implemented Algorithm~\ref{algo:generateloops}, Algorithm~\ref{algo:universalgeneralcase} and Algorithm~\ref{algo:generalinvariants} in Macaulay2~\cite{M2}.  
The prototype builds on~\cite{bayarmagnaiIssac}, and the source code is publicly available at:  

\begin{center}\scriptsize\vspace*{-0em}
\href{https://github.com/Erdenebayar2/Synthesizing_Loops.git}{\texttt{https://github.com/Erdenebayar2/Synthesizing\_Loops.git}}\vspace*{-0.5em}
\end{center}
All experiments were conducted on a laptop equipped with a 4.8~GHz Intel i7 processor, 16~GB RAM, and a 25~MB L3 cache.  
After generating polynomial systems using Algorithms~\ref{algo:generateloops} and~\ref{algo:universalgeneralcase}, we used the SMT solver \texttt{Z3}~\cite{z3solver} to search for common nonzero integer solutions.

We evaluated our approach on benchmarks from~\cite{ISSAC2023Laura,hitarth_et_al:LIPIcs.STACS.2024.41,humenberger2022LoopSynthesis}.  
The benchmark suite is publicly available at:  
\begin{center}\scriptsize\vspace*{-1em}
\href{https://github.com/Erdenebayar2/Synthesizing_Loops/tree/master/software/loops}%
{\texttt{https://github.com/Erdenebayar2/Synthesizing\_Loops/software/loops}}
\end{center}\vspace*{-0em}
\newcommand{\Idexp}{Id}

\textbf{Notations.}\; In the following tables, ``$n$'' denotes the number of program variables, ``$m$'' the number of polynomial invariants in $\gb$, and ``$d$'' their maximal degree.  
``$D$'' is the maximal degree among the polynomials in $\fb_1,\ldots,\fb_n$.  
Let $\fb_1 = (f_{1,1}, \ldots, f_{1,l_1}), \ldots, \fb_n = (f_{n,1}, \ldots, f_{n,l_n})$ be sequences of polynomials in $\C[\xb]$, and set $l = l_1 + \cdots + l_n$.  
In other words, ``$l$'' denotes the total number of generators for the update maps of the loops.

\newcommand{\Fexp}{\textbf{F}}
\newcommand{\NIexp}{\textbf{NI}}
\newcommand{\TLexp}{\textbf{TL}}

\subsection{Experimental Results}

Table~\ref{table1} reports the execution times for Algorithm~\ref{algo:generateloops} and for \texttt{Z3}, both of which are run on the polynomial systems generated by Algorithm~\ref{algo:generateloops}.  
All timings are given in seconds, with a timeout (indicated by \TLexp) of 300 seconds.  
The label ``\Fexp'' indicates that \texttt{Z3} failed to find an integer solution, while ``\NIexp'' means that no input was passed to \texttt{Z3} because Algorithm~\ref{algo:generateloops} itself timed out.  

In most cases, once Algorithm~\ref{algo:generateloops} 
terminates, \texttt{Z3} finds a common nonzero integer solution within 0.3 seconds.  
Thus, searching for integer solutions is not a bottleneck.  
We also observe that providing additional polynomial invariants generally accelerates the termination of Algorithm~\ref{algo:generateloops}.  
Overall, these results demonstrate that the main computational challenge lies in generating the polynomial systems, rather than in solving them with \texttt{Z3}.


\begin{table}[H]
    \centering\hspace*{-0.2cm}
    \scalebox{0.47}{
    \begin{tabular}{|*{4}{c|}*{5}{|c|c|}}
    \hline
        \multicolumn{4}{|c||}{Polynomial map} & \multicolumn{2}{|c||}{D=1, $l=3$} & \multicolumn{2}{|c||}{D=1, $l=4$} &  \multicolumn{2}{|c||}{D=1, $l=5$} & \multicolumn{2}{|c||}{D=2, $l=2$} & \multicolumn{2}{|c|}{D=2, $l=3$} \\ \hline
    Benchmark &$n$ & $m$& $d$ &Alg.~\ref{algo:generateloops} & \texttt{Z3} & Alg.~\ref{algo:generateloops} & \texttt{Z3} & Alg.~\ref{algo:generateloops} & \texttt{Z3} & Alg.~\ref{algo:generateloops} & \texttt{Z3} & Alg.~\ref{algo:generateloops} & \texttt{Z3}\\ \hline      
        Ex1.2~\cite{ISSAC2023Laura} & 2&1 & 4 & {0.02} & \Fexp & {31.1} & \Fexp & \TLexp & \NIexp & {0.01} & 0.06 &\TLexp&\NIexp \\ \hline
        Ex1.1~\cite{ISSAC2023Laura} & 3&2 & 3 & {0.01} & 0.06 & {0.04} & 0.06 & 4.4 & 0.2 & {0.01} & 0.06 &0.018 & 0.07 \\ \hline
        Ex1.1Ineq & 3&2 & 3 & {0.009} & 0.06 & {11.4} & 0.06 & \TLexp & \NIexp & {0.008} & 0.05 &0.03 & 0.05 \\ \hline
        Ex5.2\cite{hitarth_et_al:LIPIcs.STACS.2024.41} & 2&1 & 2 & {0.03} & 0.17 & {0.16} & \TLexp & \TLexp & \NIexp & {0.01} & 0.07 &0.13&0.07\\ \hline
      sum1~\cite{humenberger2022LoopSynthesis} & 3&2 & 2 & {0.02} & 0.06 & {0.13} & 0.06 & 1.1 & 0.06 & {0.01} & 0.05 &0.02 & 0.06 \\ \hline
        square~\cite{humenberger2022LoopSynthesis} & 2&1 & 2 & {0.01} & 0.06 & {0.5} & \TLexp & \TLexp & \NIexp & {0.01} & 0.06 &0.015 & 0.26 \\ \hline
        square\_conj~\cite{humenberger2022LoopSynthesis} & 3&2 & 2 & {0.019} & 0.06 & {0.14} & 0.09 & 0.39 & 0.06 & {0.007} & 0.05 &0.015 & 0.06 \\ \hline
        fmi1~\cite{humenberger2022LoopSynthesis} & 2&1 & 2 & {1.19} & 0.06 & {161.74} & 0.06 & \TLexp & \NIexp & {237.1} & 0.06 &\TLexp & \NIexp \\ \hline
       fmi2~\cite{humenberger2022LoopSynthesis} & 3&2 & 2 & {0.01} & 0.06 & {0.015} & 0.06 & 0.017 & 0.07 & {0.009} & 0.07 &0.01 & 0.07 \\ \hline
      fmi3~\cite{humenberger2022LoopSynthesis}& 3&2 & 2 & {0.01} & 0.06 & {0.03} & 0.05 & 0.39 & 0.09 & {0.01} & 0.06 &0.2 & 0.11 \\ \hline
       intcbrt~\cite{humenberger2022LoopSynthesis} & 3&2 & 2 & {0.25} & 0.07 & {16.6} & 0.08 & \TLexp & \NIexp & {0.01} & 0.06 &0.65 &0.17 \\ \hline
cube\_square~\cite{humenberger2022LoopSynthesis} & 3&1 & 3 & {0.01} & 0.07 & {0.018} &\TLexp& \TLexp & \NIexp & {0.15} & 0.06 &0.96 & 106 \\ \hline
Ex3 & 2&1 & 2 & {0.014} & 0.09 & {0.02} &0.1& 0.034 & 0.1 & {0.13} & 0.1 &0.07 & 0.12 \\ \hline
markov\_triples & 3&1 & 2 & {5.61} & 0.1 & 36.2 & 0.1& 134.1 & 0.6 & {1.35} & 0.1 &\TLexp & \NIexp \\ \hline
       
    \end{tabular}
    }
    
    \caption{\label{table1}Timings for Algorithm~\ref{algo:generateloops} in seconds; 
    }
\end{table}

Table~\ref{table2} presents the polynomial systems produced by Algorithm~\ref{algo:generateloops}, which define $\CL(\ab,h,\fb;\gb)$.  
Similarly, Table~\ref{table3} shows the systems produced by Algorithm~\ref{algo:universalgeneralcase}, which define $U\CL(h,\fb;\gb)$, together with the execution times for \texttt{Z3} on those systems.  
For each choice of ``$D$'' and ``$l$'', we report the number ``$s$'' of nonzero polynomials and whether the system has finitely many solutions.  
In Tables~\ref{table2} and~\ref{table3}, ``Id'' indicates that, within the given structure, the only loop satisfying the specified invariants is the identity map.

Macaulay2~\cite{M2} computes the irreducible decomposition of the generated varieties within a few seconds in most of the cases shown in Tables~\ref{table2} and~\ref{table3}.  
In many examples, the irreducible components are defined by linear equations.  
Furthermore, Table~\ref{table2} shows that even when the dimension is~0, the varieties are never empty; they always consist of finitely many points.

\newcommand{\Fin}{$<\!\!\infty$}
\newcommand{\NS}{$\#$sols\xspace}

\begin{table}[H]
    \centering\hspace*{-0.2cm}
    \scalebox{0.51}{
    \begin{tabular}{|*{4}{c|}*{5}{|c|c|}}
    \hline
        \multicolumn{4}{|c||}{Polynomial map} & \multicolumn{2}{|c||}{D=1, $l=3$} & \multicolumn{2}{|c||}{D=1, $l=4$} &  \multicolumn{2}{|c||}{D=1, $l=5$} & \multicolumn{2}{|c||}{D=2, $l=2$} & \multicolumn{2}{|c|}{D=2, $l=3$} \\ \hline
    Benchmark &$n$ & $m$& $d$ &$s$ & \NS & $s$ & \NS & $s$ & \NS & $s$ & \NS & $s$ & \NS\\ \hline      
        Ex1.2~\cite{ISSAC2023Laura} & 2&1 & 4 & 3 & $\infty$ & 4 & $\infty$ & \TLexp & \TLexp & 2 & $\infty$ &\TLexp&\TLexp \\ \hline
        Ex1.1~\cite{ISSAC2023Laura} & 3&2 & 3 & 2 & $\infty$ & 4 & $\infty$ & 6 & $\infty$ & 4 & \Fin &4& \Fin\\ \hline
        Ex1.1Ineq & 3&2 & 3 & 2 & $\infty$ & 4 &$\infty$ & \TLexp & \TLexp & 4 & \Fin &4& \Fin\\ \hline
        Ex5.2\cite{hitarth_et_al:LIPIcs.STACS.2024.41} & 2&1 & 2 & 3 & $\infty$ & 3 &$\infty$ & \TLexp & \TLexp & 3 & \Fin &3&$\infty$\\ \hline
          sum1~\cite{humenberger2022LoopSynthesis} & 3&2 & 2& 4 & \Idexp & 6 & $\infty$ & 6 & $\infty$ & 2 & \Fin &4&\Fin \\ \hline
        square~\cite{humenberger2022LoopSynthesis} & 2&1 & 2 & 2 & $\infty$ & 4 & $\infty$ & \TLexp & \TLexp & 2 & \Fin &3&$\infty$ \\ \hline
        square\_conj~\cite{humenberger2022LoopSynthesis} & 3&2 & 2 & 4 & \Fin & 6 & $\infty$ & 6 & $\infty$ & 4 & \Fin &4&\Fin \\ \hline
        fmi1~\cite{humenberger2022LoopSynthesis} & 2&1 & 2 & 0 & $\infty$ & 0& $\infty$ & \TLexp & \TLexp & 0 & $\infty$ &\TLexp&\TLexp \\ \hline
       fmi2~\cite{humenberger2022LoopSynthesis} & 3&2 & 2 & 2 & $\infty$ & 4 & $\infty$ & 4 & $\infty$ & 2 & \Fin &4& \Fin\\ \hline
      fmi3~\cite{humenberger2022LoopSynthesis} & 3&2 & 2 & 4 & \Idexp & 4 & $\infty$ & 6 & $\infty$ & 2 &\Fin &4&\Fin \\ \hline
       intcbrt~\cite{humenberger2022LoopSynthesis} & 3&2 & 2& 4& \Idexp & 4 & $\infty$ & \TLexp & \TLexp & 2 & \Fin &4&\Fin\\ \hline
cube\_square~\cite{humenberger2022LoopSynthesis} & 3&1 & 3 & 2 & $\infty$ & 3 & $\infty$ & \TLexp & \TLexp &3 &\Fin &5&\Fin \\ \hline
Ex3 & 2&1 & 2 & 3 & $\infty$ & 3 & $\infty$ & 7 & $\infty$ &3 &\Fin &7&\Fin \\ \hline
markov\_triples & 3&1 & 2 & 15 & \Fin & 15 & \Fin& 15 & \Fin & 7 & \Fin &\TLexp & \NIexp \\ \hline
       
    \end{tabular}
    }
    
    \caption{\label{table2} Data on outputs of Algorithm~\ref{algo:generateloops}  
    }
\end{table}


For every benchmark in Table~\ref{table1}, Algorithm~\ref{algo:universalgeneralcase} terminates within 0.1 seconds and \texttt{Z3} within 0.3 seconds.  
However, in most cases the polynomial equations produced by Algorithm~\ref{algo:universalgeneralcase} either admit no common solution or yield only the identity loop within the specified structure.  
This limitation arises from the fact that Algorithm~\ref{algo:universalgeneralcase} assumes the invariants to be universally inductive, whereas Algorithm~\ref{algo:generateloops} makes no such assumption.  

Accordingly, and because efficiency allows for it, in Table~\ref{table3} we use significantly larger templates for the update maps of loops than in Table~\ref{table1}.  
In particular, the number of generators in Table~\ref{table3} ranges from 6 to 121, compared to only 3 to 5 in Table~\ref{table1}.  
For these larger templates, Algorithm~\ref{algo:generateloops} fails to terminate within 300 seconds mainly because of costly Gr\"obner basis computations.
In contrast, Algorithm ~\ref{algo:universalgeneralcase} always terminates within 0.6 seconds because it relies solely on linear algebra computations. However, it generally produces a loop of a higher degree, which is more susceptible to numerical instability.

In each column of Table~\ref{table3}, the update maps are generated from all monomials up to degree~$D$.

\begin{table}[H]
    \centering\hspace*{-0.2cm}
    \scalebox{0.48}{
    \begin{tabular}{|*{4}{c|}*{3}{|c|c|c|c|}}
    \hline
        \multicolumn{4}{|c||}{Polynomial map} &  \multicolumn{4}{|c||}{D=1} &\multicolumn{4}{|c||}{D=2} &  \multicolumn{4}{|c|}{D=3} \\ \hline
    Benchmark &$n$ & $m$& $d$& $l$ & \texttt{Z3}&$s$ & \NS & $l$ & \texttt{Z3}&$s$ & \NS& $l$ & \texttt{Z3}&$s$ & \NS\\ \hline      
        Ex1.2~\cite{ISSAC2023Laura} & 2&1 & 4 & 6 &\TLexp & 15&\Fin   &12&\TLexp&45 & \Fin  & 20&\TLexp&91&\Fin \\ \hline
        Ex1.1~\cite{ISSAC2023Laura} & 3&2 & 3 & 12&0.1 &30 &\Idexp&30&0.1&119 & $\infty$ &60&0.15 &304& \TLexp\\ \hline
        Ex1.1Ineq & 3&2 & 3 & 12&0.1 &30 &\Idexp&30&0.1&119 & $\infty$ &60&0.15 &304& \TLexp\\ \hline
        Ex5.2\cite{hitarth_et_al:LIPIcs.STACS.2024.41} & 2&1 & 2 & 6&0.1 &6 &$\infty$&12&0.14&15 &$\infty$  &20& \TLexp&28&$\infty$\\ \hline
     sum1~\cite{humenberger2022LoopSynthesis} & 3&2 & 2&  12 &0.15 &14 &$\infty$&30&0.16&45& $\infty$ &60& 0.17 &104& \TLexp\\ \hline
        square~\cite{humenberger2022LoopSynthesis} & 2&1 & 2 & 6& 0.1 &6&$\infty$&12&0.1&15 &$\infty$ & 20& 0.1&28&$\infty$ \\ \hline
        square\_conj~\cite{humenberger2022LoopSynthesis} & 3&2 & 2 & 12 &0.1&14 &$\infty$ &30&0.16&45& $\infty$& 60& \TLexp & 104& \TLexp\\ \hline
        fmi1~\cite{humenberger2022LoopSynthesis} & 2&1 & 2 & 6& 0.14&6 &$\infty$&12&0.14& 15 & $\infty$ &20 & 0.15 &28&$\infty$ \\ \hline
       fmi2~\cite{humenberger2022LoopSynthesis} & 3&2 & 2 &12 &0.15  &14 & $\infty$& 30&0.15&45&$\infty$&60&0.17 &104& \TLexp\\ \hline
      fmi3~\cite{humenberger2022LoopSynthesis} & 3&2 & 2 &12  &0.2 &14 &$\infty$&30&0.15&45 & $\infty$& 60&0.17 &104& \TLexp\\ \hline
       intcbrt~\cite{humenberger2022LoopSynthesis} & 3&2 & 2&12  &0.17 &30 &\Idexp&30&0.2&119 & $\infty$ & 60&0.25 &304&\TLexp\\ \hline
cube\_square~\cite{humenberger2022LoopSynthesis} & 3&1 & 3  & 12 &0.2& 20&$\infty$ &30& 0.19& 84&\TLexp&60&0.5&220&\TLexp \\ \hline
Ex3 & 2&1 & 2 &  12& 0.21  &12 &$\infty$&24& 0.2& 30&$\infty$&40&0.21&56&$\infty$\\ \hline
markov\_triples & 3&1 & 2 & 24 &  \TLexp&40 &\Fin &60 & \TLexp&168 &\TLexp&121& \TLexp&304 & \TLexp \\ \hline  
    \end{tabular}
    }
    \caption{\label{table3} Data on outputs of Algorithm~\ref{algo:universalgeneralcase} }
\end{table}

\end{document}